\documentclass[
    a4paper,
    twocolumn,
    superscriptaddress,
    accepted=2018-05-09
        ]{quantumarticle}
\usepackage[utf8]{inputenc}
\usepackage[greek,english]{babel}
\usepackage[T1]{fontenc}
\usepackage{amsmath}

\usepackage{tikz}
\usepackage{lipsum}

\usepackage{etex}

\usepackage{graphicx}
\usepackage{epsfig,epstopdf}

\PassOptionsToPackage{hyphens}{url}
\usepackage{hyperref}
\usepackage[numbers,sort&compress]{natbib}

\newcommand*{\black}{\color{black}}

\newcommand{\uppi}{\text{\greektext{p}\latintext}} 
\newcommand{\euler}{\textnormal{{e}}} 
\newcommand{\imun}{\textnormal{{i}}}  

\usepackage{amsmath}
\usepackage{amsfonts}
\usepackage{amssymb}
\usepackage{amsthm}
\usepackage{mathtools}

\usepackage{hyphenat}

\usepackage{enumerate}

\usepackage{enumitem}

\newcommand*{\mot}{~}

\newcommand*{\quac}{\mathcal{Q_{\mathrm{ac}}}}
\newcommand{\sharpp}{$\#\textnormal{P}$}
\newcommand{\qac}{\mathcal{Q_{\mathrm{ac}}}}

\newtheorem{theorem}{Theorem}

\newtheorem{definition}[theorem]{Definition}
\newtheorem{lemma}[theorem]{Lemma}

\newtheorem{corollary}[theorem]{Corollary}

\definecolor{DarkGray}{gray}{0.25} 
\definecolor{MidGray}{gray}{0.38} 
\definecolor{NeutralGray}{gray}{0.5}
\definecolor{LightGray}{gray}{0.7}
\definecolor{lightGray}{gray}{0.85}
\definecolor{DarkRed}{rgb}{0.7,0,0}
\definecolor{DarkBlue}{rgb}{0,0,0.5}
\definecolor{SteelBlue}{rgb}{0,0.4,0.6}
\definecolor{Orange}{rgb}{0.7,0.5,0}
\definecolor{Violette}{rgb}{0.5,0,0.5}
\definecolor{Sand}{rgb}{0.84,0.8,0.55}
\definecolor{niceblue}{rgb}{0.33,0.5,0.8}
\definecolor{OliveGreen}{RGB}{0,102,102}
\definecolor{NiceGreen}{RGB}{0,153,72}
\definecolor{LightGreen}{RGB}{0,200,72}
\definecolor{newblue}{RGB}{40,210,251}
\definecolor{lightblue}{RGB}{179,231,251}
\definecolor{steelblue}{RGB}{70,130,180}
\definecolor{cred}{RGB}{179,28,28}
\definecolor{applegreen}{rgb}{0.55, 0.71, 0.0}
\usepackage{tcolorbox}

\tcbset{
    coltitle=black,fonttitle=\bfseries,
    colback=orange!7!white,colframe=gray!25!white,width=0.95\textwidth,valign=center,
    before=\par\bigskip\centering,after=\par
    }

\DeclareMathOperator{\id}{id}
\DeclareMathOperator{\tr}{tr}

\DeclareMathOperator{\poly}{poly}

\DeclareMathOperator{\diag}{diag}

\newcommand{\set}[1]{\{ #1  \}}

\newcommand{\ket}[1]{|{#1}\rangle}

\newcommand{\bra}[1]{\langle{#1}|}

\newcommand{\norm}[1]{\Vert #1\Vert}

\newcommand{\abs}[1]{\lvert #1\rvert}

\newcommand{\ee}{\mathrm{e}}
\newcommand{\ii}{\mathrm{i}}

\newcommand{\Pb}{\mathbb{P}}

\makeatletter 
\hypersetup{pdftitle = {Anticoncentration theorems for schemes showing a quantum speedup},
	     pdfauthor = {Dominik Hangleiter, Juan Bermejo-Vega, Martin Schwarz, Jens Eisert},
	     pdfsubject = {Quantum computation, quantum information},
	     pdfkeywords = {Quantum supremacy,
	     	quantum speedup, 
	     	anticoncentration, 
	     	anticoncentration,
	     	unitary designs,
	     	2-designs, 
	     	universal random circuits, 
	     	commuting circuits, 
	     	diagonal unitaries
		     }
	    }
\makeatother

\begin{document}

\date{\today}
\author{D.\ Hangleiter}
\orcid{0000-0002-4766-7967}
\email{dominik.hangleiter@fu-berlin.de}
\affiliation{Dahlem Center for Complex Quantum Systems, Freie Universit{\"a}t Berlin, 14195 Berlin, Germany}

\author{J.\ Bermejo-Vega}
\orcid{0000-0003-3727-8092}
\email{jbermejovega@gmail.com }
\affiliation{Dahlem Center for Complex Quantum Systems, Freie Universit{\"a}t Berlin, 14195 Berlin, Germany}

\author{M.\ Schwarz}
\email{martin.schwarz@gmail.com}
\affiliation{Dahlem Center for Complex Quantum Systems, Freie Universit{\"a}t Berlin, 14195 Berlin, Germany}

\author{J.\ Eisert}
\email{jense@zedat.fu-berlin.de}
\affiliation{Dahlem Center for Complex Quantum Systems, Freie Universit{\"a}t Berlin, 14195 Berlin, Germany}
\orcid{0000-0003-3033-1292}

\title{Anticoncentration theorems for schemes showing a quantum speedup}

\nohyphens{\maketitle}

\begin{abstract} One of the main milestones in quantum information science is to realise quantum devices that exhibit an exponential computational advantage over classical
ones without being universal quantum computers, a state of affairs dubbed quantum speedup, or sometimes ``quantum computational supremacy''. 
The known schemes heavily rely on mathematical assumptions that are plausible but unproven, prominently results on anticoncentration of random prescriptions. 
In this work, we aim at closing the gap by proving two anticoncentration theorems and accompanying hardness results, one for circuit-based schemes, the other for quantum quench-type schemes for quantum simulations. Compared to the few other known such results, these results give rise to a number of comparably simple, physically meaningful and resource-economical schemes showing a quantum speedup in one and two spatial dimensions. 
At the heart of the analysis are tools of unitary designs and random circuits that allow us to conclude that universal random circuits anticoncentrate as well as an embedding of known circuit-based schemes in a 2D translation-invariant architecture. 
\end{abstract}


\section{Introduction}

Realising a quantum device that computationally outperforms state-of-the art classical supercomputers for a certain task that is provably intractable classically has become a key milestone in the field of quantum simulation and computing. 
This goal is often referred to as  
``quantum (computational) supremacy'' \cite{preskill2013quantum} or quantum
speedup. 
Such a quantum speedup is not merely meant in the sense
of quantum dynamics being no longer tractable on classical supercomputers using the best known algorithms to date,
for which there is evidence already today \cite{Trotzky,MBL2D,Emergence}. 
Instead, to make sure that the inefficient classical simulation is not victim of a lack of imagination, such a quantum speedup is usually meant to refer to schemes for which the speedup can be related to a notion of computational complexity.
For a quantum speedup scheme to be physically realisable \emph{in principle} in the absence of quantum error correction, 
it is crucial that the hardness of the task is robust under physically realistic errors. 
To have any hope of realising such a scheme in the near term one would moreover wish for the resources required for an implementation of the architecture in the intractable regime to be achievable with present-day (or near-term) technology.

There are only very few quantum speedup architectures that are robust against physically realistic constant total-variation distance errors \cite{aaronson_computational_2010,bremner_average-case_2016,bremner_achieving_2017,boixo_characterizing_2016,gao_quantum_2017,Supremacy,morimae_hardness_2017,miller_quantum_2017}. 
Even fewer of those are physically realistic when it comes to an implementation in present-day technology in that they require only nearest-neighbour interactions and are feasible in the available experimental platforms such as linear optics \cite{aaronson_computational_2010}, superconducting qubits \cite{boixo_characterizing_2016,bremner_achieving_2017}, ion traps or cold atoms in optical lattices \cite{Supremacy}.
The computational task that is solved in all of these proposals is a sampling task, in particular, the task of sampling from the output
distribution of a certain random time-evolution. 
That random time evolution may take the form of a Haar-random unitary applied to a bosonic state \cite{aaronson_computational_2010}, a random circuit from a gate set \cite{boixo_characterizing_2016}, IQP circuits \cite{bremner_achieving_2017} applied to an all-zero state, or even a translation-invariant nearest-neighbour Ising Hamiltonian that is applied to a random product state \cite{Supremacy}. In addition to this discussion, there is the question to what extent 
schemes showing a quantum speedup can be certified in their correctness
\cite{SampleComplexity,AaronsonUniform,Hangleiter,Supremacy,gao_quantum_2017,miller_quantum_2017,kapourniotis_nonadaptive_2017}.

The central ingredient of all existing quantum speedup proofs is Stockmeyer's
algorithm\mot\cite{Stockmeyer} that implies a collapse of the Polynomial Hierarchy if sampling from the output distribution of the respective circuits is \sharpp-hard on average.
In order for this hardness argument to be valid, one crucially requires so-called 
anticoncentration bounds for the output probability
distribution of the respective random circuits \cite{lund_quantum_2017}. 

Indeed, it has been shown that one can efficiently classically sample from output distributions of certain circuit families, including IQP circuits, the output distribution of which concentrates on a polynomially small subset of the sample space \cite{Schwarz13_Sparse}. 
This shows that concentrated output distributions are in many relevant cases simulable, rendering anticoncentration a necessary condition for classical hardness for these cases. 

Despite of their central role in the hardness argument of quantum speedup proposals, only few proofs of anticoncentration bounds are known so far 
\cite{bremner_average-case_2016,bremner_achieving_2017,morimae_hardness_2017}. 
In all other speedup architectures -- boson sampling \cite{aaronson_computational_2010}, universal random circuits \cite{boixo_characterizing_2016}, and translation-invariant Ising models \cite{gao_quantum_2017,Supremacy} -- there exists 
none or merely numerical evidence for anticoncentration of the respective circuit families and the validity of the anticoncentration assumption needs to be conjectured. This still gives rise to plausible schemes, but in order to 
complete the program of realising quantum schemes showing a quantum speedup, these gaps must
necessarily be closed. Rigorous anticoncentration results for classically intractable circuit families are therefore both of crucial importance to corroborate the validity of those existing speedup proposals, as well as to shine light on the conditions of them coming about which are highly debated in the literature. 

In this work, we provide rigorous anticoncentration results for two types of such quantum  architectures that are at the same time not classically simulable. 
First, we show that random circuits drawn from a unitary 2-design anticoncentrate. 
We then apply this result to both show that random circuits comprised of nearest-neighbour gates that are drawn from a universal gate set containing inverses anticoncentrate in linear depth, and propose two new schemes based on this insight.
Second, we prove that the output distribution of a particular nearest-neighbour quantum quench architecture based on the time evolution of product states under certain translation-invariant Ising models anticoncentrate in constant depth. 
Thus, we consider two types of architectures tailored towards different kinds of experimental platforms in the following sense.
In platforms in which achieving large numbers of qubits is expensive and local control feasible, circuit-based schemes such as universal random circuits can be reasonably implemented. 
A paradigmatic example of such a platform might be constituted of superconducting qubits. 
In contrast, there are physically most natural settings of quantum simulators in which local control on the level of individual gates is difficult to achieve, but for which extremely large numbers of local constituents can be reached. 
Cold atoms in optical lattices, in which $10^4-10^5$ atoms are readily reachable, provide the most prominent example of such an architecture. 
Our quench-type architecture is tailored toward such settings. 
 
The application of our first result to universal random circuits complies with the intuition that due to the ballistic spread of correlations anticoncentration will generically arise in depth that scales linearly with the diameter of the system under consideration, and hence, linearly in a one-dimensional architecture \cite{boixo_characterizing_2016}. 
Still, to the best of our knowledge there is no rigorous proof for anticoncentration of universal random circuits. 
In contrast, in the light of this intuition the second result is quite surprising:
It has even been argued \cite{lund_quantum_2017} that it cannot be expected to reach anticoncentrating output distributions 
in constant depth, retaining its classical intractability. 
We conjecture the scaling of both results to be optimal in the settings considered (unstructured circuits in one dimension and highly structured circuits in two dimensions).

To prove the first result (Theorem~\ref{thm:anticoncentration}) we make use of properties of approximate unitary 2-designs and apply the Paley-Zygmund inequality. 
In applying this result to show that several schemes anticoncentrate we use the fact that all these schemes form approximate 2-designs. 
This includes, universal random circuits as in \cite{boixo_characterizing_2016}, Clifford circuits acting on input product magic states\mot\cite{jozsa_classical_2013}, and certain models of diagonal unitaries\mot\cite{nakata_generating_2014}. 
What is more, we provide new complexity-theoretic evidence for the hardness of classically simulating these random-circuit families in the approximate sampling sense. 
In doing so, we focus on universal random circuits, which attain this property already in \emph{linear} $O(n)$ depth\mot\cite{gross_evenly_2007,brandao_local_2016}. 
For this case, we derive a matching classical-hardness upper bound by proving that a broad class of 1D random quantum circuits are hard to simulate classically already in linear\mot$O(n)$ circuit depth, in the strong simulation sense (Lemma \ref{lem:ClassicalHardness}). 
The novel aspects of our work are the generality of these results and the minimal linear-depth requirements for achieving both classical hardness and anticoncentration on a 1D nearest-neighbour architecture. 
This drastically improves over prior work\mot\cite{boixo_characterizing_2016} on universal random circuits, which gave only numerical evidence for anticoncentration (in depth $O(\sqrt{n})$, and in 2D), and no matching classical-hardness depth bound. Prior to us, analogous depth results were only available for nearest-neighbour IQP circuit models supplemented by SWAP gates: for these,  Ref.\mot\cite{bremner_achieving_2017} gave $O(\sqrt{n}\log n)$ depth bounds for anticoncentration and hardness in 2D architectures, and Ref.\mot\cite{Supremacy} proved linear $O(n)$ ones in 1D layouts. 

For the second result (Theorem~\ref{lemma:IQPSampling}) and Corollaries~\ref{thm:AntiConOfQuench}--\ref{thm:QQuenchSupremacy} thereof, we use the facts that IQP circuits anticoncentrate \cite{bremner_achieving_2017} and can be implemented in a 1D architecture in linear depth \cite{Supremacy}. 
Our technical contribution is to provide an embedding of such IQP circuits in the constant-time evolution of a product state under a translation-invariant Ising model on a two-dimensional lattice. 
Via this embedding we are thus able to show the first anticoncentration bound for translation-invariant constant-depth schemes that exhibit a quantum speedup~\cite{gao_quantum_2017,Supremacy}.

This work is structured as follows: First, in Sec.~\ref{sec:definitions}, we
will introduce the formal statement of anticoncentration and show how it is used in the Stockmeyer hardness proof. In
Sec.~\ref{sec:random circuits} we will both state and prove the
anticoncentration result for approximate unitary 2-designs and then apply this result to relevant examples, most importantly, universal random circuits in linear depth. 
In Sec.~\ref{sec:ising} we will then prove the  anticoncentration result for the constant-time evolution of a random product state under a certain nearest-neighbour translation-invariant Ising Hamiltonian. 
Finally, we discuss the implications of our results in the context of the timely literature in Sec.~\ref{sec:discussion}, and conclude in Sec.~\ref{sec:conclusion}. 


\section{Preliminaries: Anticoncentration and quantum speedups}
\label{sec:definitions}

Throughout this work, we consider quantum systems consisting of $n$ qubits (with obvious generalisation to $d$-dimensional
local constituents). To start with, let us make precise, what is meant by anticoncentration of the output distribution of a unitary $U$ drawn from a certain measure $\mu$. 
We call the distribution of probabilities $\abs{\bra{x} U \ket{0}}^2$ of obtaining $x \in \set{0,1}^n$ when applying a unitary $U \in U(N)$,
$N=2^n$, to an initial state vector $\ket{0} :=  \ket{0}^{\otimes n }$ and measuring in the computational basis, the 
output distribution. 
We say that this output distribution anticoncentrates if there exist universal
constants $\alpha, \beta > 0$ such that for any  $x\in \{0,1\}^n$, the
probabilities $\abs{\langle x | U | 0 \rangle }^2$ of this unitary  anticoncentrate,
\begin{align} 
    \label{anticoncentration}
    \mathrm{Pr}_{U \sim \mu}\biggl(|\langle x | U | 0 \rangle |^2 \geq \frac{\alpha}{N} \biggr) > \beta
    \, . 
\end{align}
We can interpret this probability as the probability that an arbitrarily chosen
entry $x$ of the first column of a $\mu$-randomly chosen $U$ is larger than $\alpha/N$. 
Throughout this work, we say that a quantity $X$ is approximated by a quantity 
$\tilde{X}$ with \emph{multiplicative error} $c$ if $X/c \leq \tilde{X} \leq c
X$, with \emph{relative error} $r$ if $ (1- r) X \leq \tilde{X} \leq
(1 + r ) X$, and with \emph{additive error} $a$ if $ \norm{X - X }_* \leq a$
for some norm $\norm{\cdot}_*$.

In the commonly used proof technique for quantum speedups \cite{BremnerOld,terhal_adaptive_2004} Stockmeyer's algorithm \cite{Stockmeyer} is applied to show a collapse of the polynomial hierarchy if for an arbitrary such $x$ the amplitude $|\bra{x} U \ket{0}|^2$ is \#P-hard to approximate multiplicatively.
Anticoncentration comes into this proof when hardness is shown not up to multiplicative but up to an additive error in total-variation distance.
More specifically, to prove a quantum speedup with constant total-variation distance errors for sampling from the output distribution of a circuit family $\mathcal{F}$ using the argument developed in Refs.~\cite{aaronson_computational_2010,bremner_average-case_2016} one requires three ingredients: 
(i) The output distribution of $\mathcal{F}$ anticoncentrates in the sense of Eq.~\eqref{eq:anticoncentration}. 
(ii) post$\mathcal{F}$ = postBQP. 
By the result of Refs.\ \cite{Kuperberg15JonesPolynomial,fujii_commuting_2017} the output probabilities are then \#P-hard to approximate up to relative error $1/4$. 
(iii) The output probabilities of $\mathcal{F}$ are \#P-hard to approximate up to multiplicative errors in the \emph{average case}. This needs to be conjectured for all quantum speedup schemes\footnote{In the exact case, one can even prove average-case results for the permanent \cite{aaronson_computational_2010} and random-circuit based schemes \cite{bouland_quantum_2018}.}, preferably in terms of a universal quantity such as the imaginary-time partition function of Ising models \cite{bremner_average-case_2016}, the permanent \cite{aaronson_computational_2010}, or the Jones polynomial \cite{mann_complexity_2017}. 
Together, (ii) and (iii) permit a reduction from hardness of strong simulation up to multiplicative error to hardness of weak simulation up to an additive error using Stockmeyer's algorithm \cite{Stockmeyer} in the third level of the Polynomial Hierarchy. 
As a result, very often three conjectures need to be made when proving a quantum speedup using this technique:
\begin{enumerate}[leftmargin=18pt,label=\textnormal{C\arabic*}]
\item The Polynomial Hierarchy cannot collapse to its 3rd level \cite{aaronson2016,fortnow2005beyond,karp1980}.\label{conj:PH}
\item If it is $\#\textnormal{P}$-hard to approximate the output probability of a circuit drawn from $\mathcal{F}$ up to a constant relative error, then the same problem is $\#\textnormal{P}$-hard for a constant fraction of the instances\footnote{This conjecture may be regarded as a qubit analogue of the ``permanent-of-Gaussians'' conjecture of Ref.\ \cite{aaronson_computational_2010}.}. \label{conj:Average}
\item The output distribution of $\mathcal{F}$ anticoncentrates. \label{conj:AntiCon}
\end{enumerate} 

In the following, we will prove the anticoncentration conjecture \ref{conj:AntiCon} in the sense of equation
\eqref{anticoncentration} both for certain circuit-based schemes, in fact, those that form an approximate 2-design, and a quantum-quench architecture in the mindset of Ref.~\cite{Supremacy}.


\section{Anticoncentration of circuit-based schemes} 
\label{sec:random circuits}

In this section we will begin by introducing and proving our first result, namely, that approximate unitary 2-designs anticoncentrate, and then apply this result to three relevant examples of circuit-based schemes, most prominently, universal random circuits.

\subsection{Anticoncentration of unitary 2-designs}
\label{sec:2-designs}

Unitary $k$-designs approximate the uniform (Haar) measure on the unitary group (see App.~\ref{app:haar measure}) in the sense that the first $k$ moments of a
unitary $k$-design and the Haar measure match (exactly or approximately). 
The definition of a $k$-design is motivated by the fact that in experiments
samples from a unitary $k$-design are much easier to realise than samples from
the full Haar measure. 
In order to define the notion of a $k$ design, we need the notion of the
$k^{\text{th}}$-moment operator that acts as a unitary twirl with respect to
some measure $\mu$ on the unitary group maps on an operator.

\begin{definition}[$k^{\text{th}}$-moment operator]
Let $M_\mu^k$ be the $k$-th moment
operator on $\mathcal{L}(\mathcal{H}^{\otimes k } )$ with respect to a
distribution $\mu$ on $U(N)$, $N = 2^n = \dim \mathcal{H}$ defined as 
\begin{equation} 
\begin{split}
    X \mapsto M_\mu^k (X) & \coloneqq \mathbb{E}_\mu \left[ U^{\otimes k } X
    (U^\dagger)^{\otimes k } \right] \\
    &= \int_{U(N)}  U^{\otimes k } X
    (U^\dagger)^{\otimes k } \mu(U).
\end{split}
\end{equation} 
\end{definition}
We can now define unitary $k$-designs \cite{gross_evenly_2007,dankert_exact_2009}. 
\begin{definition}[Unitary $k$-design]
Let $\mu$ be a distribution on the unitary group $U(N)$.
Then $\mu$ is an exact unitary $k$-design if
$$M_\mu^k = M_{\mu_{\rm Haar}}^k \, .  \label{eq:exact design} $$ 
\end{definition}

In all of what follows, we will need to relax this notion to the notion of an \emph{approximate} unitary $k$-design. 
In such a definition we can allow for both relative and additive errors
on the equality \eqref{eq:exact design} \cite{brandao_local_2016,onorati_mixing_2017}: 
\begin{definition}[Approximate unitary $k$-designs]
Let $\mu$ be a distribution on the unitary group $U(N)$.
Then $\mu$ is  
\begin{enumerate}
	\item  an additive $\epsilon$-approximate unitary $k$-design if 
        $$ \Vert M_\mu^k - M_{\mu_{\rm Haar}}^k \Vert_\diamond \le \epsilon \,
        , $$
   \item a relative $\epsilon$-approximate unitary $k$-design if 
       $$(1-\epsilon) M_{\mu_{\rm Haar}}^k \leq M_\mu^k \leq (1 + \epsilon)
       M_{\mu_{\rm Haar}}^k \, . $$ 
\end{enumerate}
\end{definition}

Since the former definition is much more common in the literature, let us remark that the two definitions are closely related via the following Lemma of
Ref.~\cite{brandao_local_2016} in which, however, a factor of the dimension enters. 
\begin{lemma}[Additive and relative approximate designs]
	If $\mu$ is a relative $\epsilon$-approximate unitary $k$-design then $\norm{M_\mu^k - M_{\mu_{\rm Haar}}^k}_\diamond \leq 2\epsilon$. 
	Conversely, if  $\norm{M_\mu^k - M_{\mu_{\rm Haar}}^k}_\diamond \leq\epsilon$, then $\mu$ is a relative $\epsilon N^{2k}$-approximate unitary $k$-design. 
\end{lemma}

We are now ready to state our first anticoncentration result on unitary 2-designs.

\begin{theorem}[Anticoncentration of unitary 2-designs]
    \label{thm:anticoncentration} 
    Let $\mu$ be a relative $\epsilon$-approximate unitary 2-design on
    the group $U(N)$.  
    Then the output probabilities $\abs{\bra{x} U \ket{0} }^2$ for $x \in
    \set{0,1}^n$ of a $\mu$-random unitary $U \in U(N)$ anticoncentrate in the sense that for $0 \leq \alpha \leq 1$
    \begin{equation}
        \Pb_{\substack{U \sim \mu}} \left( \abs{\bra{x} U \ket{0} }^2 > \frac{\alpha
(1-\epsilon) }{N} \right)  \geq  \frac{(1-\alpha)^2(1-\epsilon)^2}{2(1+\epsilon)} \, . 
    \label{eq:anticoncentration} 
\end{equation} 
\end{theorem} 


We point out that Theorem~\ref{thm:anticoncentration} also holds in exactly the
same way for relative $\epsilon$-approximate state 2-designs. 
This is a weaker condition than the unitary design condition since any (approximate) unitary 2-design generates an (approximate) state 2-design via application to an arbitrary reference state. 
Also note that the fact that $\mu$ is a relative $\epsilon$-approximate
1-design (cf.\ App.~\ref{app:k k-1 designs}) is crucial for the bound
\eqref{eq:anticoncentration} to become non-trivial. 
If instead $\mu$ was an additive design the lower bound would asymptotically tend to zero as $1/N$ and hence not stay larger than a constant. 
However, the 1-design condition holds even exactly for many distributions $\mu$, although for the higher moments it may only hold approximately.

\begin{proof}[Proof of Theorem \ref{thm:anticoncentration}]
    Our proof of the anticoncentration bound \eqref{eq:anticoncentration}
    has two steps and relies on two ingredients: In the first step, we prove 
    anticoncentration of a single \emph{but fixed} entry of Haar random
    unitaries. To this end we make use of the \emph{Paley-Zygmund
    inequality} and an explicit expression of the distribution of
    matrix elements of Haar-random unitaries. 
    In the second step, we extend this result to full anticoncentration of all
    output probabilities in the sense of equation \eqref{eq:anticoncentration}.

    The Paley-Zygmund inequality is a lower-bound analogue of Markov-type tail
    bounds and can be stated as follows. 
    If $Z \geq 0$ is a random variable with finite variance, and if $0 \le \alpha \le 1$
    \begin{align}
        \Pb ( Z > \alpha \mathbb{E}[Z] ) \geq (1- \alpha)^2 \frac{\mathbb{E}[Z]^2}{\mathbb{E}[Z^2]}\, . 
    \end{align}
    That is, it lower bounds the 
    probability that a positive random variable is small in terms of its mean and variance.

    Now let $\mu$ be a relative $\epsilon$-approximate unitary 2-design.
    Then for $l = 2, 4$, it holds that 
    \begin{equation} 
    \begin{split} ( 1 -   \epsilon) 
        \mathbb{E}_{U \sim \rm Haar}&  \left[ \abs{ \bra{a} U
\ket{b} } ^l  \right ] 
        \leq  
\mathbb{E}_{U \sim \mu} \left[ \abs{ \bra{a} U
        \ket{b} } ^l \right ] \\
        & \leq ( 1 +   \epsilon) 
        \mathbb{E}_{U \sim \rm Haar} \left[ \abs{ \bra{a} U
        \ket{b} } ^l  \right ]  \, . 
    \end{split}
\end{equation} 
    This is due to the fact that for any unitary $k$-design $\mu$
    the expectation value of an arbitrary polynomial $P$ of degree
    $2$ in the matrix elements of both $U$ and $U^\dagger$ over 
    $\mu$ equals the same expectation value but taken over the Haar measure up
    to a relative error $\epsilon>0$ \cite{harrow_random_2009}.
    To see this,
    observe that averaging a monomial in the matrix elements of $U$ 
    over the $k$-design $\mu$ can be expressed
    as $\langle i_1, \ldots, i_k | M_\mu^k (|j_1, \ldots, j_k \rangle \langle j_1',
    \ldots, j_k' |) | i_1', \ldots, i_k' \rangle$. 
    Hence, if $ M_\mu^k = M_{\rm
    Haar}^k$, ``then any polynomial of degree $k$ in the matrix elements of $U$ 
    will have the same expectation over both distributions''\mot\cite{harrow_random_2009}. 
This gives rise to 
    \begin{equation} 
    \begin{split}
        \Pb_{U \sim \mu}  &( \abs{\bra{x} U \ket{0} } ^2  > \alpha
        (1-\epsilon) \mathbb{E}_{U \sim \rm Haar} [\abs{\bra{x} U
        \ket{0} } ^2 ] )\\&  \geq \Pb_{U \sim \mu} ( \abs{\bra{x} U
            \ket{0} } ^2   > \alpha \mathbb{E}_{U \sim \mu} [\abs{\bra{x} U \ket{0} } ^2 ] ) \\
        & \geq (1- \alpha)^2
        \frac{\mathbb{E}_{U \sim \mu}[\abs{\bra{x} U \ket{0} } ^2
        ]^2}{\mathbb{E}_{U \sim \mu}[\abs{\bra{x} U \ket{0} } ^4]}\\
         & \geq (1- \alpha)^2
        \frac{(1-\epsilon)^2 \mathbb{E}_{U \sim \rm Haar}[\abs{\bra{x} 
        U \ket{0} } ^2 ]^2}{(1+\epsilon)\mathbb{E}_{U \sim \rm Haar}[\abs{\bra{x} U \ket{0} } ^4]}.
        \end{split}
\end{equation}

    \begin{lemma}[Marginal output distribution]
    \label{lemma:output distribution}
        The distribution of the marginal output probabilities $p = \abs{\bra{x}
        U\ket{0} } ^2$ of Haar random unitaries $U$ and arbitrary but fixed $x$ is
       given by
        \begin{equation}
            \label{eq:haar distribution}
            P_{\rm Haar }(p) = (N- 1)(1 - p )^{N-2} \xrightarrow{N \gg 1} N \exp(-N p ).  
        \end{equation}
        In particular, $P_{\rm Haar }$'s  first and second moments  are given by 
        \begin{align}
        \mathbb{E}_{\rm {Haar}}[p]   = \frac{1}{N} ,\quad
        \mathbb{E}_{\rm {Haar}}[p^2]   = \frac{2}{N(N + 1 )} .\label{eq:firstandsecondmoment}
        \end{align}
    \end{lemma} 
    We prove this lemma in App.~\ref{matrix elements}. 
    Inserting the expressions Eqs.~\eqref{eq:firstandsecondmoment} for $\mathbb{E}_{U \sim \rm Haar}
    [\abs{\bra{x} U \ket{0} } ^2 ]$ and $\mathbb{E}_{U \sim \rm Haar} [\abs{\bra{x} U \ket{0} } ^4]$, we find
    \begin{equation} 
    \begin{split}
        P_{U \sim \mu} &\left( \abs{\bra{x} U \ket{0} } ^2  > \frac{\alpha (1-\epsilon)
}{N} \right)\\ 
                & \geq (1-\alpha)^2 \frac{N ( N + 1 ) }{2 N^2
    }\frac{(1-\epsilon)^2}{(1+\epsilon)}\geq (1-\alpha)^2
        \frac{(1-\epsilon)^2}{2(1+\epsilon)}\, ,  \notag \\
    \end{split}
\end{equation} 
which completes the proof. 
\end{proof} 

Note that the moments \eqref{eq:firstandsecondmoment}
can alternatively be obtained for both state and unitary 2-designs exploiting Schur-Weyl duality. 
This yields an explicit expression of the $k^{\text{th}}$-moment
operators $M_\mu^k(X)$ as the projector onto the span of the symmetric group on $k$ tensor copies of the Hilbert space $\mathcal{H}$. Moreover, the moments of the output probabilities of state 2-designs are also given by \eqref{eq:firstandsecondmoment}. 
This can be seen similarly using the fact that the expectation value over a state $k$-design is given by the projection onto the $k$-partite symmetric subspace of $\mathcal{H}^{\otimes k }$ \cite{zhu_clifford_2016}.
We note that the output 
distribution \eqref{eq:haar distribution} of a Haar random unitary asymptotically approaches the exponential (Porter-Thomas) distribution. 
This behaviour has already been observed numerically in many different contexts involving pseudo-random operators \cite{gross_evenly_2007,emerson_pseudo-random_2003}, non-adaptive measurement-based quantum computation \cite{brown_quantum_2008-1}, and universal random circuits \cite{boixo_characterizing_2016}. 

\subsection{Applications: a ``recipe'' for quantum speedups} 
\label{sec:applications}

Our anticoncentration theorem \ref{thm:anticoncentration} for approximate unitary (and state) 2-designs leads to a generic ``recipe'' for the identification of quantum circuit families and input states that are hard to simulate classically under plausible complexity-theoretic conjectures, building upon the approach of Refs.~\cite{aaronson_computational_2010,bremner_average-case_2016}. 
The strategy goes in three steps parallel to ingredients (i-iii) of the  proof based on Stockmeyer's algorithm introduced in Sec.~\ref{sec:definitions}.
In the following, we apply this general strategy to a few examples of random circuits, most prominently, universal random circuits. 

\begin{figure} 
    \includegraphics[width=\columnwidth]{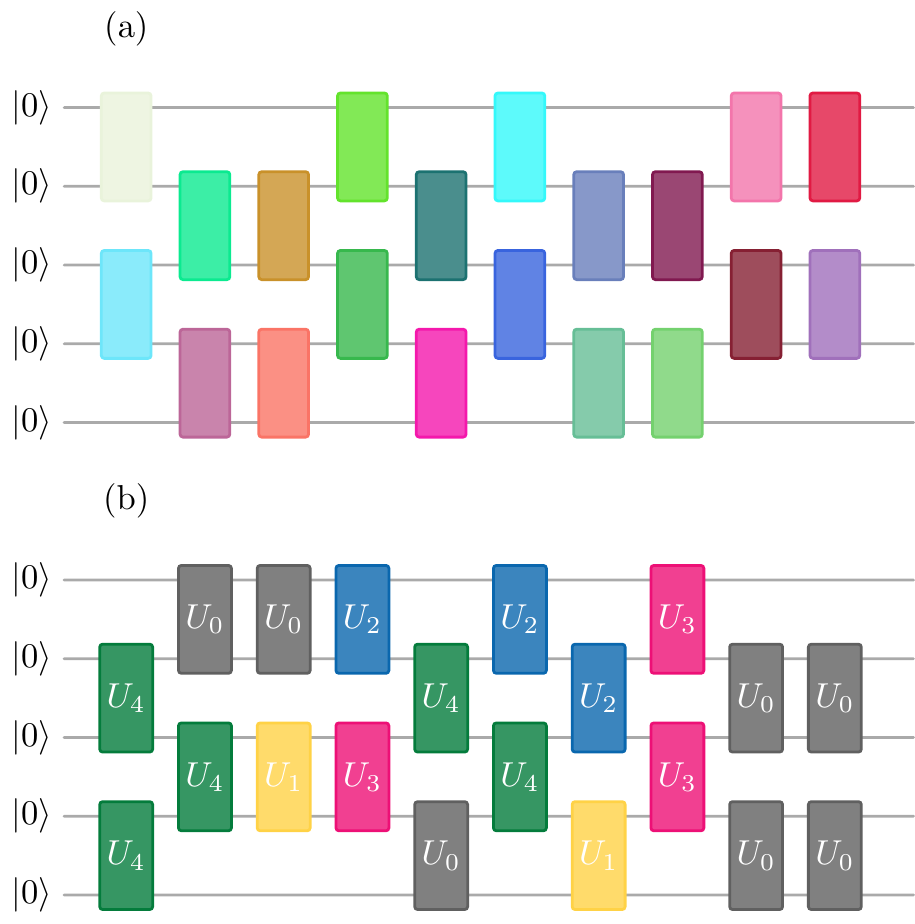}
    \caption{
        Layout of the parallel random circuit families. In each step either the
        even or odd configuration of parallel two-qubit unitaries is applied
        with probability $1/2$. Every two-qubit gate is chosen from the
        respective measure on $U(4)$ -- (a) the Haar measure, (b) the uniform
        distribution on the gate set $G$. 
        Here we depict a five-qubit random instance of depth 10 where in (a)
        the colour choice represents different gates, and in (b) the gate set
        consists
        of 5 two-qubit unitaries $G = \{U_0, U_2, \ldots, U_4\}$. 
        \label{circuit families}
    }
\end{figure} 

\paragraph{Universal random circuits.}

The first example that we also focus on are random quantum circuits constructed from
single- and two-qubit gates, most prominently, the gate set $\mathcal{G}_{\text{BIS}} =
\{ CZ, H, \sqrt{X}, \sqrt{Y}, T \} $ studied by Boixo \emph{et
al.}~\cite{boixo_characterizing_2016,neill_blueprint_2017}. 
In presenting this example we put particular emphasis on the \emph{circuit depth} required to reach a scheme that shows a provable quantum speedup. 
As a first step (i) of the general strategy, the following Corollary establishes that the output distribution of a random circuit formed in a particular fashion from $\mathcal{G}_{\text{BIS}}$ anticoncentrates. 
This holds already in \emph{linear depth}. 

\begin{corollary}[Universal random circuits anticoncentrate]
    \label{circuits corollary} 
    The output probabilities of universal random circuits in one dimension from the following
    two circuit families (illustrated in Fig.~\ref{circuit families}) anticoncentrate in a depth that scales as $O(n \log(1/\epsilon))$ in the sense of Eq.~\eqref{eq:anticoncentration}. 
\begin{itemize}
    \item \emph{Parallel local random circuits}: 
        In each step either the unitary $U_{1,2 } \otimes U_{3,4} \otimes \cdots \otimes U_{n-1,n}$ or the
        unitary $U_{2,3 } \otimes U_{4,5} \otimes \cdots \otimes U_{n -  2,n-1}$ is applied (each with probability $1/2$), with $U_{j,j+1}$ independent
        unitaries drawn from the Haar measure on $U(4)$. (This assumes $n$ is even.)

    \item \emph{Universal gate sets}: Let $G \coloneqq \set{g_i}_{i=1}^m$ with
        each $g_i \in U(4)$ be a universal gate set containing inverses with elements composed of algebraic identities, 
        i.e., a gate set $G$ such
        that the group generated by $G$ is dense in $U(4)$ and satisfying 
        $g_i \in G
        \Rightarrow g_i^{-1} \in G$. 
        In each step either the unitary $U_{1,2 } \otimes U_{3,4} \otimes \cdots \otimes U_{n-1,n}$ or the
        unitary $U_{2,3 } \otimes U_{4,5} \otimes \cdots \otimes U_{n-2,n-1}$ is applied (each with probability $1/2$), with $U_{j,j+1}$ independent
        unitaries drawn uniformly from $G$.

\end{itemize}
\end{corollary} 

\begin{proof}[Proof of Corollary \ref{circuits corollary}]
    The central ingredient of our proof of Corollary \ref{circuits corollary}
    is the result of Ref.\ \cite{brandao_local_2016}. 
    There, the authors show that the two random circuit families are
    relative $\epsilon$-approximate unitary $k$-designs on $U(2^n)$ in
    depth $\poly(k) \cdot O(n  \log(1/\epsilon))$ (Corollary 6 and 7
    in Ref.\ \cite{brandao_local_2016}).

    Hence, in particular, these random circuits are relative
    $\epsilon$-approximate unitary 2-designs in depth $O(n 
    \log(1/\epsilon))$, i.e., linear in the number of qubits and logarithmic in
    $1/\epsilon$. 
    Applying Theorem \ref{thm:anticoncentration} to the output probabilities
    $ \abs{\bra{x} \mathcal{C} \ket{0}  }^2$ of a random
    circuit $\mathcal{C}$ applied to an initial all-zero state yields the claimed anticoncentration bound for the output probabilities of such circuits. 
\end{proof} 

To prove a quantum speedup using the Stockmeyer technique the second required ingredient is \#P-hardness of strong classical simulation of the output probabilities (ii).
Indeed, since the gate set $\mathcal{G}_{\text{BIS}}$ is universal, the post$\mathcal{F} =$postBQP connection is immediate. 
Boixo \emph{et al.}~\cite{boixo_characterizing_2016} moreover showed
that the output probabilities can be expressed in terms of the imaginary-time
partition function of a random Ising model, suggesting that the average-case
conjecture for random circuits is a natural one (iii). 
It remains to be shown that random universal circuits are both not classically strongly simulable and anticoncentrate in \emph{linear depth} in a one-dimensional setting. 
Lemma~\ref{lem:ClassicalHardness} (below) establishes this is indeed the case for a large class of finite gate sets with efficiently-computable matrix entries (so that they cannot artificially encode solutions to hard problems).
It is an open question whether this can be improved to square-root depth in a two-dimensional setting such as that of Refs.~\cite{boixo_characterizing_2016,neill_blueprint_2017}.


Given two $O(1)$-local gate sets $A$ and $B$, we say that $A$ \emph{exactly synthesises} $B$ if every gate $V\in B$ can
be exactly implemented via a polynomial-time computable constant-size circuit of gates in $A$. 
\begin{lemma}[Hardness of strong classical simulation]
\label{lem:ClassicalHardness} Let $G$ be any finite
    universal gate set with algebraic efficiently computable matrix entries
    that can exactly synthesise either the $\{\euler^{\imun\frac{\uppi}{8}
    X}$, $\euler^{\imun \frac{\uppi}{4}  X \otimes
    X},\textnormal{SWAP}\}$ or  $\{\euler^{\imun \frac{\uppi}{8}  X
    \otimes X},\textnormal{SWAP}\}$. 
    Then, approximating the output probabilities
of $O(n)$-depth circuits of $G$ nearest-neighbour gates in one dimension up to relative error $1/4$ is \#\textnormal{P}-hard.
\end{lemma} 
Let us highlight that Lemma~\ref{lem:ClassicalHardness} applies to many
well-studied universal gate sets, including $\mathcal{G}_{\text{BIS}}$, the ubiquitous Clifford+$T$
\cite{boykin_new_2000},
Hadamard+controlled-$\sqrt{Z}$ \cite{kitaev2002classical}, Hadamard+Toffoli
\cite{Shi:2003:BTC:2011508.2011515,Paetznick13UFTQC_Transversal} and others
\cite{Shor:1996:FQC:874062.875509,Knillquant-ph/9610011,Knill365}. 
Interestingly, Lemma~\ref{lem:ClassicalHardness} holds
also for non-universal gate sets, though the latter may not always anticoncentrate.    
We now prove Lemma~\ref{lem:ClassicalHardness}. 

\begin{proof}[Proof of Lemma~\ref{lem:ClassicalHardness}]
We begin by showing that both given target gate sets can exactly implement subgroups of the 2-qubit dense IQP circuits of Ref.\ \cite{bremner_average-case_2016}. 
    Specifically, the first gate set gives us the group $\mathcal{G}_1$ generated by $\exp\left(\imun\frac{\uppi}{8}X_i\right)$, $\exp\left(\imun\frac{\uppi}{4}X_iX_j\right)$ gates acting on a complete graph, while the second gives the group generated by arbitrary long range $\exp\left(\imun\frac{\uppi}{8}X_iX_j\right)$ gates. 
    In both cases, long range interactions are obtained via the available SWAPs. 

Next, we show that, like the circuits in Ref.\ \cite{bremner_average-case_2016},
    both $\mathcal{G}_1$ and $\mathcal{G}_2$ are universal under post-selection.
    Indeed, both can adaptively implement a single-qubit Hadamard via gate
    teleportation \cite{GateTeleportation} (see also
    \cite{RaussendorfPhysRevA.68.022312,BremnerOld}), and non-adaptively, if we
    can post-select. 
    The claim follows from the universality of known gate sets\mot\cite{boykin_new_2000,kitaev2002classical}.

Last, due to Refs.\ \cite{goldberg_complexity_2014,fujii_commuting_2017},
the output probabilities of post-selected universal quantum circuits  are $\#$P-hard to approximate up to multiplicative error $\sqrt{2}$ (relative error\mot$1/4$). The previous fact implies that this holds for the dense IQP circuits in $\mathcal{G}_1$ and $\mathcal{G}_2$. Furthermore,  $n$-qubit dense IQP circuit can be  exactly implemented in $O(n)$ depth on a 1D nearest-neighbour architecture using  SWAP gates \mot\cite[Lemma~6]{Supremacy}. It follows that the output probabilities of linear-depth circuits in $\mathcal{G}_1$ or $\mathcal{G}_2$ are $\#$P-hard to approximate. This readily extends to any circuit family that can exactly synthesise either of the former, since this process only introduces a constant depth overhead.
\end{proof}

We do not know whether Lemma~\ref{lem:ClassicalHardness}
extends to arbitrary gate sets since applying some Solovay-Kitaev type gate synthesis algorithms
\cite{Kitaev97_SolovayKitaevEarly,NielsenChuang,kitaev2002classical} should introduce a polynomial overhead factor in depth. 
This is because due to Chernoff-Hoeffding's bound \#P-hard-to-approximate quantum probabilities need to be (at least) super-polynomially small, for otherwise they could be inferred in quantum polynomial time by mere sampling, which is not believed possible\mot\cite{Bennet97StrenghtsWeaknnesses_QC,Aaronson-ProcRS-2005}. 
To approximate such small probabilities via the Solovay-Kitaev algorithm 
requires $\Omega(n^\alpha)$ overhead for some $\alpha > 0$ assuming the
counting exponential time hypothesis \cite{dell_exponential_2014}.   
These issues are closely related to the open question of whether or not the
power of post-selected quantum circuits is gate set independent given some
$\tilde{O}(n^\alpha)$ depth bound\mot\cite{Kuperberg15JonesPolynomial}.

\begin{table*}
\centering
{\small
\begin{tabular}{|c|c|| c|c|c|}
\hline
Circuit family & Input state & (State) 2-design & Worst-case hardness & Average-case conjecture in   
\\
$\mathcal{F}$& $\ket{\psi_0} $ &  property& (post$\mathcal{F} = $ postBQP) & terms of universal quantity\\ \hline \hline

$\mathcal{G}_{\text{BIS}}$ & $\ket{0}^{\otimes n}$ & \cite{brandao_local_2016} & \cite{boixo_characterizing_2016,boykin_new_2000} & Ising part.\ func./Jones polyn.\\
\hline
Diagonal unitaries & $\ket{+}^{\otimes n}$ & \cite{nakata_generating_2014} & \cite{fujii_commuting_2017,Kuperberg15JonesPolynomial} & Ising partition function\\\hline
Clifford circuits & $(T \ket{0})^{\otimes n} $ & \cite{dankert_exact_2009}& \cite{jozsa_classical_2013} & Ising partition function\\
\hline 
\end{tabular}
\caption{ Examples of random circuit families that exhibit a provable quantum speedup up to total-variation distance errors. 
	\label{tab:recipe}
}
}
\end{table*}


\paragraph{Commuting circuits.}

As a second example, we consider circuits of diagonal unitaries composed of controlled-phase type one- and two-qubit gates of the form $\diag(1, 1, 1, \ee^{\ii \phi}) $, and an input state $\ket{+}^{\otimes n}$. 
By the result of Ref.~\cite{nakata_generating_2014} this gate set yields a state 2-design if the phases are picked from discrete sets ($\{0,\pi\}$ for the two-qubit gates, and $\{ 0, 2\pi/3, 4\pi/3 \}$ for the single qubit gates), and thus satisfies anticoncentration in the sense of Eq.~\ref{eq:anticoncentration} (i).
Adding the $S$-gate to the gate set and measuring all qubits in the $X$-basis we obtain post$\mathcal{F} =$ postBQP by Refs.~\cite{fujii_commuting_2017,bremner_average-case_2016} as IQP circuits are an instance of diagonal unitaries (ii). 
Here, we have used the fact that adding the $S$-gate and post-selection gives us access to the universal gate set Clifford + $\pi/12$ \cite{cui_universal_2015,bocharov_efficient_2015}. 
Again, the average-case conjecture can be phrased in terms of an Ising partition function (iii). Last, the circuits can be implemented in linear depth if either long-range interactions or nearest-neighbour SWAPs are allowed\mot\cite{Supremacy}.


\paragraph{Clifford circuits with product-state inputs.} A similar argument can be applied to Clifford circuits which are known to be an exact 2-design \cite{dankert_exact_2009,cleve_near-linear_2015} applied to magic input states. 
By the result of Ref.~\cite{koenig_how_2014} an arbitrary element of the Clifford group in $2^n$ dimensions can be decomposed into $O(n^3)$ elementary Clifford gates. 
The result of Ref.~\cite{cleve_near-linear_2015} even achieves an exact 2-design using only quasi-linearly many one- and two-qubit Clifford gates. 
The post$\mathcal{F} = $postBQP for this case is due to Ref.~\cite{jozsa_classical_2013}. 
We summarise these examples in Table~\ref{tab:recipe}.


\section{Anticoncentration of quenched many-body dynamics}\label{sec:ising}

In this section, we investigate anticoncentration of a particular type of quantum simulation scheme exhibiting a quantum speedup based on the architecture recently introduced in Ref.\ \cite{Supremacy} (specifically, architectures I-II therein). 
These implement quenched (constant-time) dynamical evolutions \cite{1408.5148,ngupta_Silva_Vengalattore_2011} under many-body Ising Hamiltonians on the square lattice whose graph we denote by $\mathcal{L}=(V,E)$. 
More specifically, we consider a particular variant of such a setting and prove both anticoncentration for its output distribution and the hardness of strongly classically simulating it. 
In virtue of the Stockmeyer-type argument presented in Sec.~\ref{sec:definitions} this gives rise to a new quantum speedup result for this architecture.

\subsection{A new quantum quench architecture}
\label{sec:qac}

Let us start by reviewing the idea of the quench architectures introduced in Ref.~\cite{Supremacy}. 
There, the computation in the circuit model amounts to, first, preparing a  product state vector $\ket{\psi_\beta} = \bigotimes_{i \in V} (\ket{0} + \ee^ {\ii \beta_i} \ket{1})/\sqrt{2}$ with $\beta_i$ chosen randomly from a finite set of angles; 
second, implementing a constant-time evolution $U \coloneqq \euler^{-\imun H}$ under a nearest-neighbour translation-invariant Ising Hamiltonian
\begin{align}\label{eq:HardIsingModels}
H:= \sum_{(i,j)\in E}J_{i,j} Z_iZ_j - \sum_{i\in V} h_i Z_i,
\end{align} 
and, third, measuring all qubits in the $X$ basis. 
Ref.\ \cite{Supremacy} proved that quantum simulations of this form cannot be efficiently classically sampled from up to constant total-variation distance assuming variants of the average-case and anticoncentration conjecture.  As supporting evidence for the anticoncentration conjecture, Ref.~\cite{Supremacy} built a link to the anticoncentration of certain families of universal random circuits and provided  numerical data. 

We note that these architectures are closely related to that of Gao \emph{et al.} \cite{gao_quantum_2017}. The similarities and differences are spelled out in Ref.~\cite{Supremacy}.

We now introduce a new variant of such a quantum quench architecture, named $\quac$ and illustrated in Fig.\ \ref{fig:QuenchArchitecture} that produces provably hard-to-approximate anticoncentrated distributions which cannot be classically sampled from if an average-case conjecture holds and the Polynomial Hierarchy does not collapse. 
The architecture picks a uniformly-random input product state from a finite family $\mathcal{S}_{\mathrm{ac}}=\{\ket{\psi_\beta}\}_\beta$ and lets it evolve under a nearest-neighbour translation-invariant Hamiltonian $H_\mathrm{ac}$ (defined below). 

Specifically, the architecture uses $n(m):=m(2m+1)$ qubits, arranged in an $m$-row  $(2m+1)$-column square lattice.  
Boundary qubits on even-rows are initialised on  $\ket{0}$ or $\ket{1}$ uniformly at random. 
The remaining ones are divided in two groups, named ``blue'' and ``yellow'', using a lattice 2-colouring that places no blue qubit on the top-left and top-right columns. Blue qubits  are  initialised on $(\ket{0}+\ket{1})/\sqrt{2}$;
yellow ones  on  $\euler^{-\imun k_i\uppi  Z_i/8}\ket{+}$   with uniformly-random $k_i\in\{0,1,2,3\}$. 

Next, the prepared state evolves under a translation-invariant Ising Hamiltonian $H_{ac}$. Letting $[i,j]$ denote the qubit on the $i$-th row and $j$-th column lattice (in left-to-right top-to-bottom order), the latter reads
\begin{align}\label{eq:AntiConHamiltonian}
H_{\textrm{ac}}&{=}\sum_{\substack{i \leq k,\, j \leq l \\([i,j],[k,l])\in E}} \tfrac{\uppi}{4} \delta_{i,j}^{k,l} Z_{[i,j]}Z_{[k,l]} - \sum_{ v \in V} \tfrac{\uppi}{4} \mathrm{deg}_\mathcal{I}(v) Z_v,\\
\delta_{i,j}^{k,l}&{:=}\begin{cases}
0 &\textnormal{if $(i \neq k)\wedge\left(j=0\bmod 4\right) \wedge ([i,j]$ is blue)},\notag\\
0 &\textnormal{if $(i\neq k)\wedge\left(j=2\bmod 4\right) \wedge ([i,j]$ is yellow)},\notag\\
1 & \textnormal{otherwise.}
\end{cases}
\end{align}
Above, $\delta_{i,j}^{k,l}$ is the indicator function of the edge set $E_\mathcal{I}$ of a (4,2)-periodic interaction sub-lattice $\mathcal{I}=(V,E_\mathcal{I}) \subset \mathcal{L}$ (Fig.\mot\ref{fig:QuenchArchitecture}) and $\mathrm{deg}_\mathcal{I}(v)$ is the degree of $v\in V$ in $\mathcal{I}$. It is easily seen  that $\mathcal{I}$ is a brickwork pattern of 2-square-cells with closed boundaries (Fig.\mot\ref{fig:QuenchArchitecture}). The net effect of the dynamics is to implement a controlled-Z gate on every pair of  neighbouring qubits in $\mathcal{I}$ before all qubits are measured in the $X$ basis.

In what follows the central quantity will be the probability distribution defined by the probabilities
\begin{equation}\label{eq:QAC_Probability}
q_\mathrm{ac}(x,\beta)\coloneqq \left\lvert\langle x| H^{\otimes n}\euler^{-\imun H_\mathrm{ac}}\ket{\psi_\beta}\right\rvert^2 / |\mathcal{S}_{\mathrm{ac}}|
\end{equation}
of measuring the outcomes $x$ after picking $\ket{\psi_\beta}$ and evolving it under $H_{\mathrm{ac}}$ for unit time.
We also let $x_R$ be $x$'s sub-string of rightmost column's outcomes in the square lattice, and $x_L:=x-x_R$ be its set-theoretic complement.

\begin{figure}[t]
\centering
\includegraphics[width=0.8\linewidth]{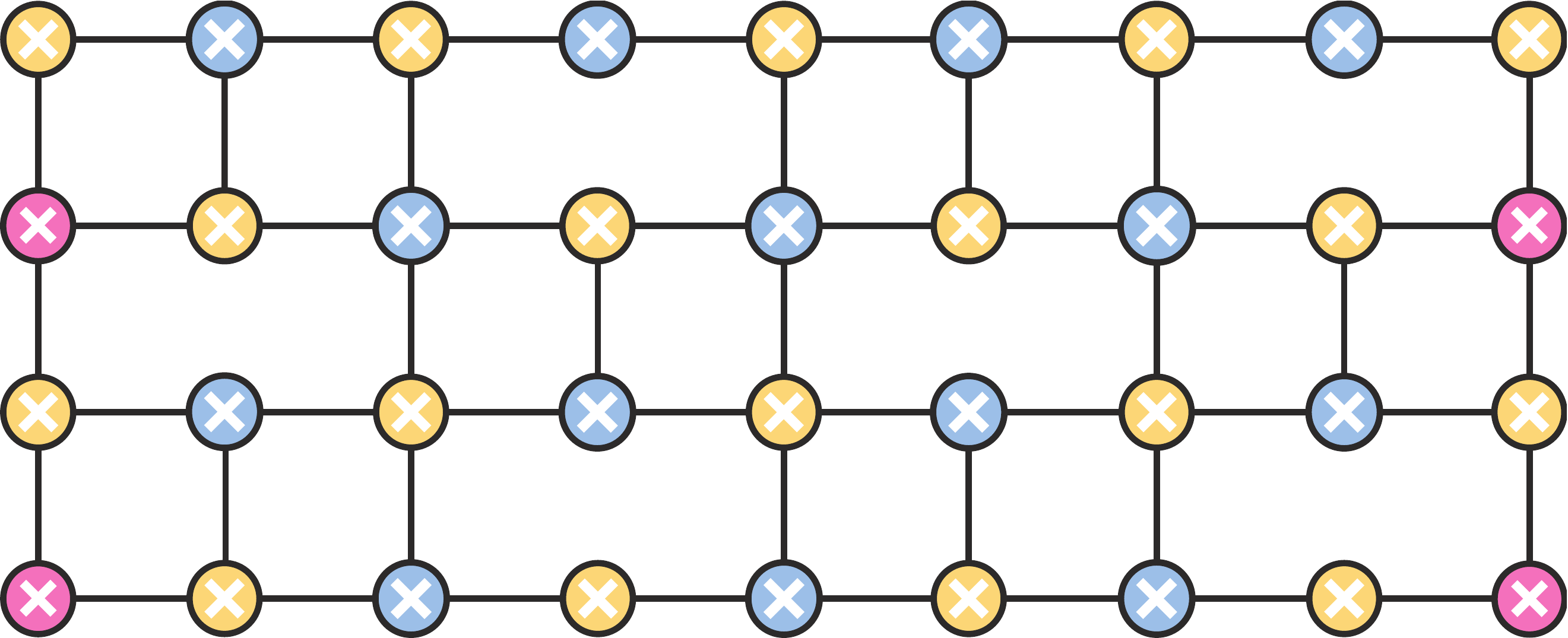}
\caption{Quantum quench architecture. Circles denote qubits, lines denote interactions. Blue sites are initialised on $|+\rangle$; pink ones on $\ket{0}$ or $\ket{1}$ at random; yellow ones  on  $\euler^{-\imun k_i\uppi  Z/8}\ket{+}$   with  random $k_i\in\{0,1,2,3\}$. The Hamiltonian evolution (\ref{eq:AntiConHamiltonian}) implements  CZ gates on connected qubit pairs. Qubits are measured 
in the $X$ basis.}
\color{lightgray} \hrule \black
\label{fig:QuenchArchitecture}.
\vspace{-15pt}
\end{figure}


\subsection{Anticoncentration and classical hardness of $\mathcal{Q}_{\mathrm{ac}}$}

Our second main result from which anticoncentration and classical hardness follow, shows that this quantum quench architecture is closely related to the ``dense'' IQP circuit family of Ref.\ \cite{bremner_average-case_2016} consisting of circuits of $\euler^{\imun \theta_i \uppi X_i}$, $\euler^{\imun \theta_{i,j} \uppi X_iX_j}$ gates acting at arbitrary pairs of qubits, with $\theta_i, \theta_{i,j}$ chosen fully and uniformly at random from $\{k \uppi/8: k \in \{0,\ldots, 7\}\}$. 
Dense IQP circuits form a commutative finite group under multiplication $\mathcal{G}_{\textrm{IQP}}$ with Haar measure $\mu_{\textrm{IQP}}$. 
The implementation of a dense IQP circuit proposed in \cite{bremner_average-case_2016} requires a fully-connected architecture, and the enactment of $\Theta(m^2)$ long-range gates for $m$-qubits in average. 
Here, we show that our constant-depth nearest-neighbour  architecture $\quac$ implements \emph{exact} sampling over dense IQP circuits with a linear $(2m+1)$ overhead-factor in qubit number.

\begin{theorem}[Quantum quench architecture]\label{lemma:IQPSampling}
For $n(m)=m(2m+1)$ qubits, the output probability distribution $q_\mathrm{ac}$ of architecture $\quac$ fulfils
\begin{equation}
q_\mathrm{ac}(x_L|\beta)=\frac{1}{2^{n-m}},\,\,\,\,q_\mathrm{ac}(x_R|x_L,\beta)=|\langle x_R|V_{x_L,\beta}\ket{0}|^2\notag 
\end{equation}
for  some $x_L,\beta$-dependent $m$-qubit dense IQP circuit  $V_{x_L,\beta}\in\mathcal{G}_{\textnormal{IQP}}$ such that $\mathrm{Pr}_{(x_L,\beta)\sim q_\mathrm{ac}}(V_{x_L,\beta})=\mu_{\textnormal{IQP}}(V_{x_L,\beta})$.
\end{theorem}

Before we turn to proving this theorem, let us state the two Corollaries important for us, namely, that both the anticoncentration result and the simulability result proven for dense IQP circuits in Ref.~\cite{bremner_average-case_2016} carry over to the quench architecture $\mathcal{Q}_{\text{ac}}$.

\begin{corollary}[Anticoncentration from quenched dynamics]\label{thm:AntiConOfQuench} The distribution $q_\mathrm{ac}$ (\ref{eq:QAC_Probability}) of the  quench architecture $\quac$ described below satisfies $q_\mathrm{ac}(\beta)=1/|\mathcal{S}_{\mathrm{ac}}|$ and
\begin{equation}
\mathrm{Pr}_{\beta\sim q_\mathrm{ac}} \left( q_\mathrm{ac}(x|\beta)\geq \frac{1}{2N}\right)\geq \frac{1}{12}.\label{eq:AntiConcentrationOfQuench}
\end{equation}
\end{corollary}

\begin{corollary}[Hardness of approximation]\label{lemma:SharpPhardness}
Approximating either $q_\mathrm{ac}(x,\beta)$ or $q_\mathrm{ac}(x|\beta)$ up to relative error 1/4 is \#\textnormal{\textnormal{P}}-hard.
\end{corollary}
Corollary \ref{thm:AntiConOfQuench} proves anticoncentration for $\quac$'s  output probabilities, while Corollary \ref{lemma:SharpPhardness} shows that  the latter are \#\textnormal{P}-hard to approximate. As an application of these two technical statements, a quantum-speedup result follows from a Stockmeyer argument. 
This result is based on the average case conjecture \ref{conj:Ising}:
\begin{enumerate}[leftmargin=18pt,label=\textnormal{C\arabic**}]
	\item Let $H_v \coloneqq H+\sum_{i\in V} v_i Z_i$ be the random Ising model derived from (\ref{eq:HardIsingModels}) by adding uniformly random on-site fields $v_i Z_i$ with random angles $v_i = (\beta_i + x_i)/2$, where $\beta_i$ are the phases of the input state $\ket{\psi_\beta}$ and $x_i$ are the measurement outcomes. The conjecture states that, if its $\#\textnormal{P}$-hard to approximate the imaginary temperature partition function  $\mathrm{tr}\left[\exp(\imun H_v) \right]$  up to a constant relative error, then the same problem is $\#\textnormal{P}$-hard for a constant fraction of the instances---intuitively, because random Ising models have no visible structure making this problem easier in average  (see also Refs.\ \cite{bremner_average-case_2016,boixo_characterizing_2016,bremner_achieving_2017}).\label{conj:Ising}
\end{enumerate}

\begin{corollary}[Intractability of classical sampling]\label{thm:QQuenchSupremacy}
If conjectures \ref{conj:PH}-\ref{conj:Ising} hold, then a classical computer cannot 
sample from the outcome distribution of architecture $\quac$ up to $\ell_1$-error  $1/192$  in time 
$O(\mathrm{poly}(n))$.
\end{corollary}
The significance of Corollary \ref{thm:QQuenchSupremacy} is that, as shown below, architecture $\qac$ defines a resource-wise plausible, certifiable experiment for demonstrating a quantum speedup with  experimental demands competitive to those in Ref.\ \cite{Supremacy}.  
However, unlike the latter, the speedup of $\qac$ relies only on a natural average-case hardness conjecture about Ising models and a Polynomial Hierarchy collapse.  Hence, we believe this corollary should help in the analysis of quantum speedups in near-term  quantum devices.

\subsection{Proof of Theorem \ref{lemma:IQPSampling} and its Corollaries}

\begin{proof}[Proof of Theorem~\ref{lemma:IQPSampling}]
We make use of two circuit gadgets, named ``odd'' and ``even'', illustrated next,
\begin{figure}[h]
	\label{fig:Gadgets}
	\centering
	\vspace{-5pt}
	\begin{tabular}{c   c}
	\textsf{Odd gadget} & \textsf{Even gadget}\\
	\includegraphics[width=0.45\linewidth]{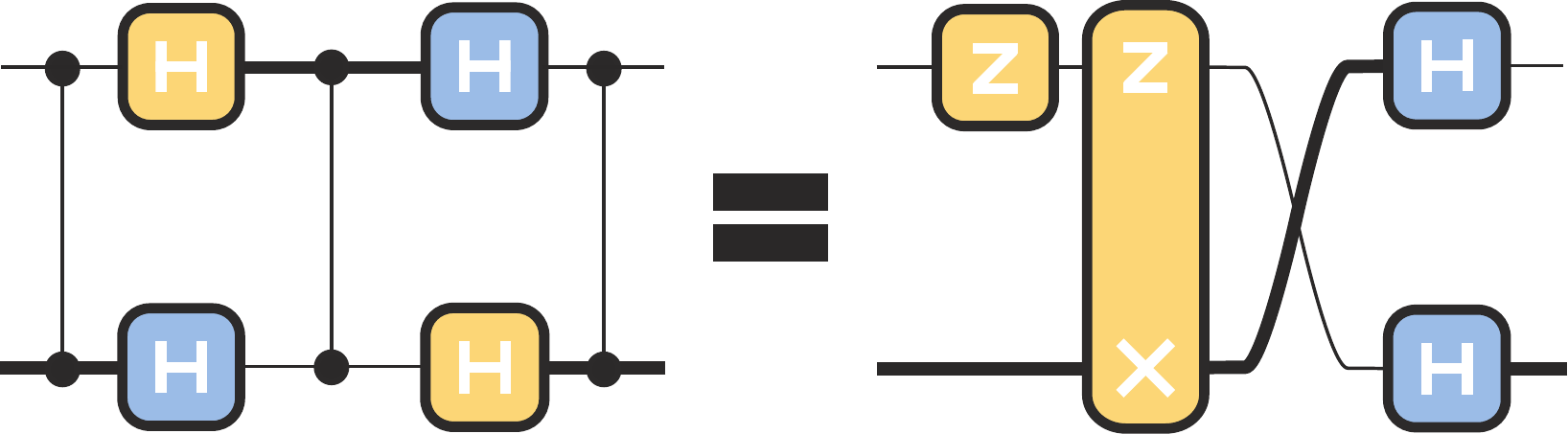}
	&
	\includegraphics[width=0.45\linewidth]{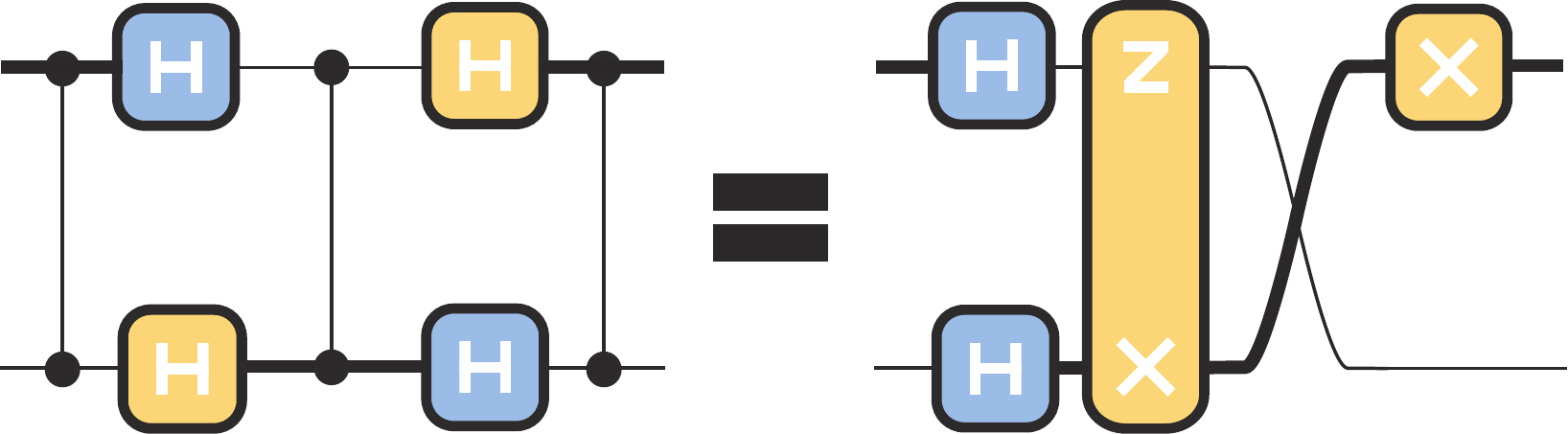}
	\end{tabular}
	\vspace{-10pt}
\end{figure}\\
which represent quantum circuit identities modulo terminal Pauli operator corrections. Vertical links represent CZ gates. Crossing qubit lines perform SWAP gates. Blue ``H'' blocks implement Hadamard gates preceded by  uniformly-random $Z_i^{x},x\in\{0,1\}$, 
single-qubit gates. Yellow ``H'' blocks,  $H_i\euler^{-\imun \uppi k Z_i/8}Z_i^{x'}$  gates with uniformly-random $0\leq k \leq 3$, $x'\in\{0,1\}$, where $\euler^{-\imun k\uppi  Z_i/8} Z_i^{x'}$ is  a uniformly-random power of $\euler^{-\imun \uppi  Z_i/8}$ up to a global phase since the latter has order 8 and $Z_i\propto\euler^{-\imun 4 \uppi  Z_i/8}$. Analogously, yellow ``Z'' (resp.\ ``ZX'' blocks) perform  uniformly-random powers of $\euler^{-\imun \frac{\uppi }{8} Z_i}$ (resp.\ $\euler^{-\imun \frac{\uppi }{8} Z_i X_{i+1}}$). The correctness of the identities is easily verified using the stabiliser formalism \cite{NielsenChuang}. 
Pauli corrections correspond to ``by-product'' $Z$s in blue blocks, which we can propagate to the end of the circuit by flipping some of the $\euler^{-\imun k\uppi  Z_i/8}$ gates' angles in yellow blocks, which leaves them  invariant.

We next show that the computation carried out by $\qac$ is equivalent to a 1D circuit of our odd and even gadgets composed in a brickwork layout. We begin by reminding the reader of the  properties of $X$-teleportation circuits \cite{GateTeleportation}, namely, that given an $(r+1)$-qubit state vector $\ket{\psi}\ket{+}$, the effect of measuring the $i$-th qubit  of $\ket{\psi}$ in the $D^\dagger XD$ basis after entangling it with $\ket{+}$ via a CZ gate is, first, to produce a uniformly-random bit $x$; second, teleport the value of the former qubit onto the latter; and, third, implement a single-qubit gate $H_{r+1} Z_{r+1}^x D_{r+1}$ on site $r+1$. Next, note that pink sites in Fig.\ \ref{fig:QuenchArchitecture} can be eliminated from the lattice by introducing  uniformly-random simultaneous $Z_i$ rotations on their neighbouring qubits. Combining these three facts and using induction, we obtain that $\qac$ can be simulated exactly by an algorithm that first generates a uniformly-random classical bit-string $x_L\in\{0,1\}^{n-m}$ and then draws $x_R$ from the output of the following network of random 1D nearest-neighbour quantum gates,
\begin{figure}[h]
\centering
\includegraphics[width=1\linewidth]{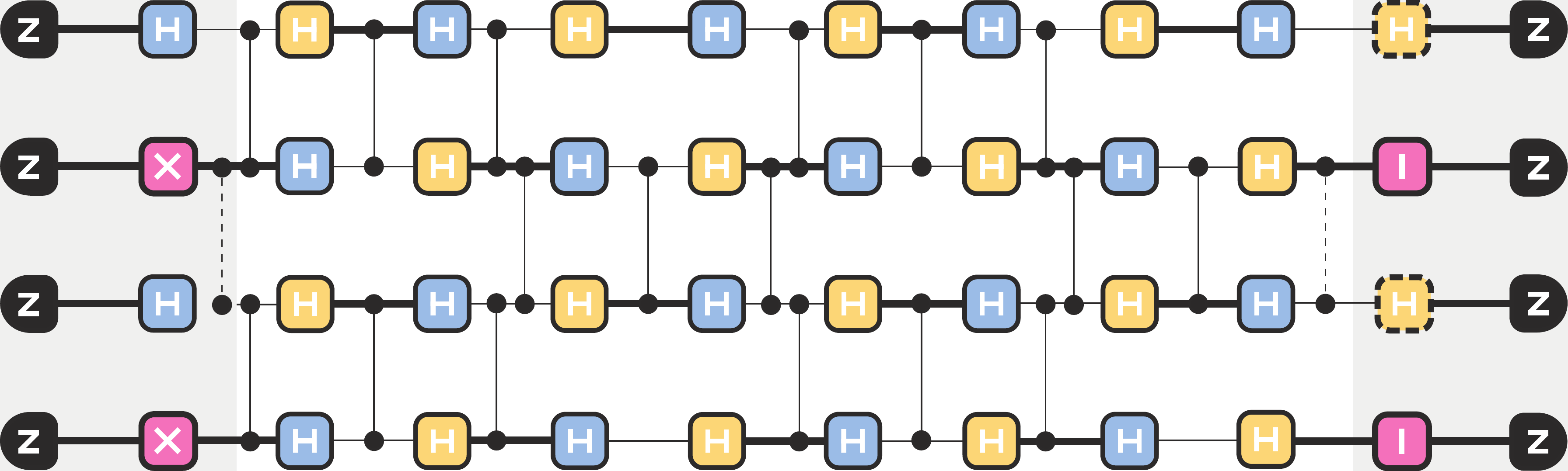}\label{fig:LogicalQuenchArchitecture}
\end{figure}\\
which we draw for $m=4$ and explicate next. The ``bulk'' of this  network (white area) contains an $m$-layered  brickwork layout of odd and even gadgets with boundaries connected by pairs of blue and yellow blocks. Blue/yellow blocks act as before. $n-m$ out of these are placed in the bulk; their associated random $Z_i^{x_{L_i}}$ gates originally correspond to the by-product rotations introduced via $X$-teleportation, and are activated by the algorithm depending distinct bits values of $x_L$. Qubits are initialised on $\ket{0}$, followed by a ``blue'' Hadamard (resp.\ a ``pink'' uniformly random $\{I_i, X_i\}$) gate on odd (resp.\ even) rows.  Even qubit lines are measured on the $Z$ basis (preceded by ``pink'' identity gates in the figure); and odd ones  on the $X$ basis preceded by a $\euler^{\imun \frac{\uppi k}{8}Z_i}$  gate. Straight-line random blocks are mutually uncorrelated (terminal ``dashed'' ones are not). Dashed CZs  are ``gauge gates'' that can be included or removed from by inserting  CNOT gates at predetermined input/output locations and reinterpreting the measurement outcomes. As before, we assume $H$ block's by-product Pauli operators are w.l.o.g.\ conjugated to the end.

Next, we apply  our odd/even gadgets to the bulk of our network to rewrite the full quantum circuit in an $m$-layered brickwork normal form 
\begin{figure}[h]
\centering
\includegraphics[width=.8\linewidth]{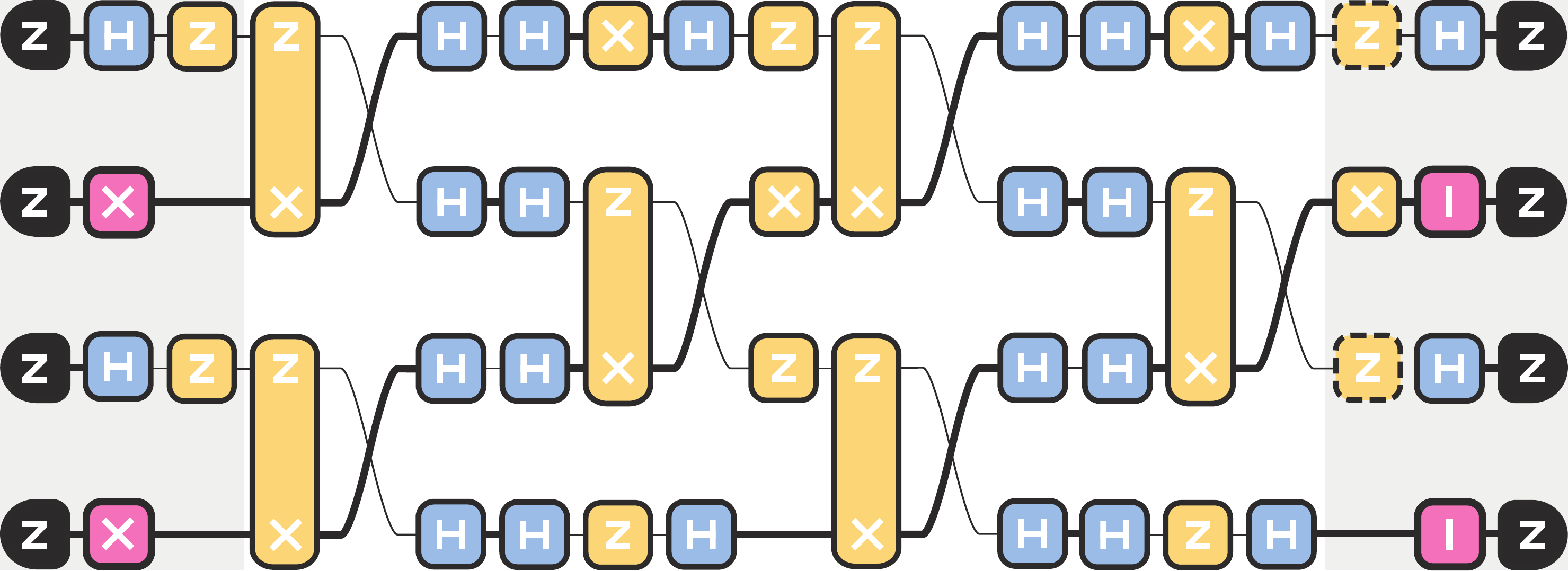} \label{fig:DenseIQPGauged}
\end{figure}\\
where odd layers execute random gates of the form 
\begin{equation}
\prod_{\textrm{odd }i}\left[H_{i}H_{i+1}\textrm{SWAP}_{i,i+1}  \euler^{-\imun \frac{\uppi a_i}{8}  Z_i X_{i+1}} \euler^{-\imun \frac{\uppi b_i}{8} Z_i}\right],\notag
\end{equation}
with  $a_i,b_i{\in}\mathbb{Z}_8$, followed by random-gate even layers of the form
\begin{equation}
\prod_{\textrm{even }i}\left[\euler^{-\imun \frac{\uppi d_i}{8} X_i}  \textrm{SWAP}_{i,i+1}  \euler^{-\imun \frac{\uppi c_i}{8}  Z_i X_{i+1}} H_{i-1}H_{i}\right], \notag
\end{equation} 
where $c_i,d_i{\in}\mathbb{Z}_8$ and we define $H_j=Z_j=X_j=\textrm{SWAP}_{k,k+1}=1$ for $j,k<1$ and $j,k+1> n$. Trailing Hadamard gates in odd layers cancel out  with their counterparts in even-layers. By a  parity-counting argument, it follows  that SWAP gates move qubits initially on odd (resp.\ even) rows travel down (resp.\ up) the circuit; the latter first undergo $Z$-type (resp.\ $X$-type) interactions, meet an odd number of $H$ gates when they reach the bottom (resp.\ top) qubit line, and then undergo the opposite process. By propagating all Hadamards in the full circuit to the measurement step, we are left only with a bulk of $n$ brickwork layers of uniformly-random $\euler^{-\imun \frac{\uppi a_i}{8}  X_i X_{i+1}}$, $\euler^{-\imun \frac{\uppi b_i}{8} X_i}$ and SWAPs, and some additional IQP gates and random Pauli by-products in the preparation/measurement steps. It was shown in Ref.\ \cite{Supremacy} that all pairs of qubits in a bulk circuit of the given form meet exactly once, hence, the  network implements exact sampling  over dense IQP circuits (crucially, due to their lack of temporal structure). Furthermore the remaining gates are either also dense IQP gates, which  leave the Haar measure $\mu_{\textnormal{IQP}}$ invariant,  or terminal Pauli $Z$ gates, which do not affect the final measurement statistics.
\end{proof}

We now exploit the mapping in Lemma \ref{lemma:IQPSampling} between $\quac$'s  and IQP circuits' output statistics to prove Corollaries \ref{thm:AntiConOfQuench} and \ref{thm:QQuenchSupremacy}.
\begin{proof}[Proof of Corollary~\ref{thm:AntiConOfQuench}] 
Recall  that  $m$-qubit dense IQP circuits fulfil
\begin{equation} 
\mathrm{Pr}_{V\sim\mu_{\textnormal{IQP}}}\left[|\langle x|V \ket{0}|^2\geq \tfrac{1}{2^{m+1}} \right]
 \geq \tfrac{1}{12},\forall x{\in}\{0,1\}^m.
\end{equation}
Since $V_{x_L,\beta}$ is drawn according to $\mu_\textrm{IQP}$ in Lemma \ref{lemma:IQPSampling}, we get 
\begin{equation} 
\mathrm{Pr}_{(x_L,\beta)\sim q_\mathrm{ac}}\left[
q_\mathrm{ac}(x_R |x_L,\beta)\geq \tfrac{1}{2^{m+1}}\right]\geq  \tfrac{1}{12}.
\end{equation}
Since $q_\mathrm{ac}(x_L|\beta)={1}/{2^{n-m}}$, we derive (\ref{eq:AntiConcentrationOfQuench}). Last, $q_\mathrm{ac}(\beta)={1}/{|\mathcal{S}_\mathrm{ac}|}$ by definition.
\end{proof}

\begin{proof}[Proof of Corollary \ref{thm:QQuenchSupremacy}]
The proof of Corollary \ref{thm:QQuenchSupremacy} is analogous to those of 
Ref.\ \cite[Theorem 1]{Supremacy} and Ref.\mot\cite[Theorem 7]{bremner_average-case_2016},
    noting that $X$-measurements on qubits prepared in states $\ket{0}$ or $\ket{1}$ in $\quac$ are equivalent to the $Z$-measurements on qubits prepared in the $\ket{+}$-state of architecture III in Ref.\mot\cite{Supremacy}. 
Then, the same argument as in Ref.~\cite{Supremacy} shows that the output probabilities $q_\mathrm{ac}(x_L,x_R|\beta)$ are proportional to an Ising partition function as in conjecture \ref{conj:Ising}. 

The only remaining difference with the proof of Theorem 1 in Ref.~\cite{Supremacy} is that we employ a different anticoncentration bound. Here, we use Eq.\mot(\ref{eq:AntiConcentrationOfQuench}) of Theorem \ref{thm:AntiConOfQuench}, which is the same bound used in Ref.\mot\cite[Theorem 7]{bremner_average-case_2016}. 
As a result, we obtain a bound of 1/192 for the allowed sampling error identical to that of Theorem 7 in \cite{bremner_average-case_2016}.
\end{proof}

\section{Implications and discussion} 
\label{sec:discussion} 

We now discuss the implications of our two main anticoncentration results and discuss possible improvements in both settings.

First, we conjecture that the linear-circuit-depth scaling in our anticoncentration result for universal random circuits in one dimension is optimal. 
Indeed, on the one hand, this result is in agreement with the intuition that anticoncentration arises as soon as correlations have spread across the entire system, a process that occurs ballistically and thus scales with the diameter of the system. On the other hand, for one-dimensional random universal circuits to be intractable classically, the depth needs to be polynomial in the number of qubits. 
Hence, our result only leaves room for a sub-linear improvement, since for circuits of poly-logarithmic depth there is a quasi-polynomial time classical simulation based on matrix-product states.
However, as is argued in Refs.~\cite{bremner_achieving_2017,lund_quantum_2017}, it would seem counter-intuitive that one can achieve sub-linear depth.
Indeed, standard tensor network contraction techniques would allow any
output probabilities of a circuit of depth $t$ in one dimension to be computed in a time scaling as $O(2^{t})$ 
\cite{Jozsa2006}. 
Hence, if the depth $t$ as a function of $n$ required for the classical hardness of generic circuits could be brought down to sub-linear, this would violate the counting exponential time hypothesis  \cite{CountingConjecture} and is therefore considered highly unlikely.

Second, we highlight that  the anticoncentration result for the two-dimensional quenched-dynamics setting provably achieves the optimal asymptotic scaling of depth, namely, constant in the number of qubits. 
This is due to the highly specific structure of the dynamical evolution and not believed to hold in an approach that relies on sampling random gates such as Refs.\ \cite{boixo_characterizing_2016,bremner_achieving_2017,aaronson_complexity-theoretic_2016}.  
Indeed, in such settings a scaling as $\Theta(\sqrt{n})$ is expected to be necessary and sufficient for an average-case hardness result and hence for anticoncentration. Again, this is due to the ballistic spreading of correlations in the system. Last, the discussed connections between 2D quenches and  one-dimensional random circuits lead us to conjecture that the required lattice width in our result, $m\times (2m+1)\in O(m^2)$, is also asymptotically optimal.

\section{Conclusion} 
\label{sec:conclusion} 

In summary, we have presented two anticoncentration theorems for quantum 
speedup schemes that are based on simple nearest-neighbour interactions and hence realisable with
plausible physical architectures, filling a significant gap in the literature.
We contrast the anticoncentration property of random circuits in one dimension that are sampled from a universal gate set with anticoncentration of the output distribution of quenched constant-time evolution of product states under translation-invariant nearest-neighbour Ising models. 
In the former setting the depth required to achieve classical hardness and at the same time anticoncentration of the output distribution scales with the diameter of the system size. 
In the latter setting a similar hardness and anticoncentration result is achieved after evolution for constant time. 
We argue that both results are optimal for the respective setting. We hope that this kind of endeavour significantly
contributes to the quest of realising quantum devices that outperform classical supercomputers, equipped with
strong complexity-theoretic claims.

\section{Acknowledgements}
\label{acknowledgements}
We are grateful to Richard Kueng for pointing us to the application of our result to diagonal unitary circuits.
Moreover, we thank Richard Kueng and Emilio Onorati for insightful
discussions and comments on the draft, Andreas Elben for discussions on Haar
random matrices, Tomoyuki Morimae for comments on the draft, and the EU Horizon 2020 (640800 AQuS), the ERC (TAQ), the Templeton Foundation,
the DFG (CRC 183, EI 519/7-1, EI 519/14-1), and the Alexander-von-Humboldt Foundation for support.

\bibliographystyle{apsrev4-1}
\bibliography{Random_unitaries}


\begin{appendix} 
\section{Some facts on random matrix theory, the Haar measure, and unitary designs}

\subsection{The Haar measure} 
\label{app:haar measure}

In this appendix, we give a precise definition of the Haar measure on the unitary group. 
To do so, let us first define a Radon measure. 

\begin{definition}[Radon measure]
	Let $(X,\mathcal{T})$ be a topological space and $\mathcal{B}$ its Borel algebra. A \emph{Radon measure} on $X$ is a measure $\mu : \mathcal{B} \rightarrow [0,+\infty]$ such that
	\begin{enumerate}[label=\roman*]
		\item for any compact set $K \subset X, \ \mu(K)< \infty$
		\item for any $B \in \mathcal{B}, \ \mu(B)=\inf \{ \mu(V) : B \subset V \text{ and                    } V \text{ open} \}$
		\item for any open set $V \subset X, \ \mu(V)=\sup \{ \mu(K) : K \subset V \text{                  and } K \text{ compact} \}$.
	\end{enumerate}
\end{definition}


\begin{definition}[Haar measure on the unitary group]
	The Haar measure is the unique (up to a strictly positive scalar factor) Radon measure which is non-zero on non-empty open sets and is left- and right-translation invariant, i.e.
	\begin{equation}
	\mu_{\rm Haar} (U) > 0 \quad \text{ for any non-empty open set } U \subset \mathcal{U}  
	\end{equation}
	and
\begin{equation}
	\mu_{\rm Haar} (B)=\mu_{\rm Haar} (u B)=\mu_{\rm Haar}(B u)
\end{equation}
for any $u \in \mathcal{U}$ and Borel set $B$ of $\mathcal{U}$, where the left- and right-translate of $B$ with respect to $u$ is given by
\begin{equation}
	uB=\set{u\, b : b \in B} \qquad \text{ and} \qquad 	Bu=\set{b\, u : b \in B}.
\end{equation}
\end{definition}


\subsection{Random matrix ensembles} 
\label{app:random matrices}

For the calculation of the distribution matrix elements of Haar-random unitaries it is instructive to introduce a few important ensembles of random matrices.
In this appendix we do so from a rather hands-on perspective. 

\begin{itemize}
    \item G$(N)$ (Ginibre Ensemble): The set of matrices $Z$ with complex Gaussian
        entries. 
        
        G$(N)$ is characterised by the measure $\mathrm{d}\mu_G(Z) \coloneqq
        \pi^{-N^2} \exp(-\tr(Z^\dagger Z) \mathrm{d} Z $, i.e., each individual
        entry $z_{i,j}$ is distributed as $\exp(-\abs{z_{i,j}}^2)/\pi$. 
    \item GUE$(N)$ (Gaussian Unitary Ensemble): The set of $N \times N$ Hermitian matrices with
        complex Gaussian entries, i.e., $H \in \mathrm{GUE} \Leftrightarrow H
        = D + R + R^\dagger$, where $D$ is a diagonal matrix with real Gaussian
        entries and $R$ is an upper triangular matrix with complex Gaussian
        entries. 

        GUE($N$) is characterised by the measure $\mathrm{d}\mu_{GUE} = {Z_{\text{GUE}(N)}}^{-1} \exp(-{
            {N} \mathrm{tr} (H^2)}/2 ) \mathrm{d} H$ on the space of Hermitian matrices. 
    \item CUE$(N)$ (Circular Unitary Ensemble): The set of Haar-random $N
        \times N $ unitary matrices.

        CUE$(N)$ is characterised by the Haar measure
        $\mathrm{d}\mu_{\mathrm{Haar}}$.
\end{itemize} 

    All of $\mathrm{d}\mu_{GUE}, \mathrm{d}\mu_G$, and $
    \mathrm{d}\mu_{\mathrm{Haar}}$ are left-
and right invariant under the action of $U(N)$. 
There are two ways of constructing Haar-random matrices. 
\begin{enumerate} 
    \item Draw a Gaussian matrix $Z \in \text{G}(N)$, and perform the unique QR
        decomposition such that $Z = QR$, with an orthogonal matrix $Q$ and 
        $R$ is required to have positive diagonal entries. 
        Setting $U = Q$, yields a Haar-random unitary \cite{ozols_how_2009,
        mezzadri_how_2006}.

    \item Draw a GUE matrix $Z \in$ GUE. Since $Z$ is Hermitian, the
        eigenvectors $v_i, i = 1, \ldots, N$ of $Z$ are orthonormal. 
        Multiplying each eigenvector $v_i$ by a random phase $\ee^{\phi_i}$ we
        can construct a Haar-random unitary matrix $U = (\ee^{\phi_1} v_1 \,
        \ee^{\phi_2} v_2 \, \cdots \, \ee^{\phi_n} v_N)$ writing those
        eigenvectors into the columns of $U$ \cite{weinstein_matrix_2005}. 
\end{enumerate} 
\medskip

\subsection{Unitary designs}
\label{app:k k-1 designs}
It is a simple exercise to show that if $\mu$ is a unitary $k$-design, all up to the $k^{\rm th}$ moments of $\mu$ equal the moments of the Haar measure. 
\begin{lemma}[$k-1$ designs from $k$ designs]
    \label{k and k-1 designs}
    Let $\mu$ be a distribution on the unitary group $U(N)$ that is an exact
    unitary $k$-design. Then $\mu$ is also a $(k-1)$-design.
\end{lemma}
\begin{proof}[Proof ]

Let $\mu$ be a unitary $k$ design. That means that it holds 
\begin{align}
    \mathbb{E}_\mu\left[ U ^{\otimes k } X(U^\dagger)^{\otimes k } \right]
    = \mathbb{E}_{\rm Haar} \left[ U ^{\otimes k } X(U^\dagger)^{\otimes k } \right]
\end{align}
for all operators $X$ acting on $\mathcal{L}(\mathcal{H}^{\otimes k})$.
Choose $X = Y \otimes \id$ with $Y$ being an arbitrary operator on
$\mathcal{L}(\mathcal{H}^{\otimes k-1})$. Then 
\begin{align}
    \mathbb{E}_\mu\left[ U ^{\otimes k -1} Y(U^\dagger)^{\otimes k -1 } \right]
    = \mathbb{E}_{\mu} \left[ U ^{\otimes k } X(U^\dagger)^{\otimes k } \right]
\end{align}
i.e., $\mu$ is a unitary $(k-1)$-design. 
\end{proof}
\begin{corollary}[Approximate $k-1$ designs from approximate $k$ designs]
    Let $\mu$ be an (additive or relative) approximate unitary $k$-design. Then $\mu$ is also an
    approximate unitary $(k-1)$-design, i.e., 
    \begin{align}
        \norm{M_\mu^k - M_{\rm Haar}^k}_\diamond \leq \epsilon \Rightarrow
        \norm{M_\mu^{k-1} - M_{\rm Haar}^{k-1}}_\diamond \leq \epsilon 
    \end{align}
    and likewise for relative errors. 
\end{corollary}
\smallskip

\section{Matrix elements of Haar-random unitaries} 
\label{matrix elements} 

Let us now derive the distribution of the amplitudes $\abs{\bra{a} U
\ket{b}}^2$ of the matrix elements a Haar-random unitary $U$
\cite{zyczkowski_random_1994,pozniak_composed_1998,weinstein_matrix_2005,haake_quantum_2010}.
To this end we apply knowledge about the distribution of entries of eigenvectors of GUE matrices and their relation to Haar-random unitaries (see App.~\ref{app:random matrices}). 
We follow Ref.~\cite{haake_quantum_2010}, Chapter 4.9. 

The eigenvectors $v_i$ of a given operator $H \in$ GUE$(N)$ have $N$ complex
components $c_k$
and unit norm $\norm{v_i}_2 = 1$. Since every eigenvector can be unitarily
transformed into an arbitrary vector of unit norm, the only invariant
characteristic of those eigenvectors is the norm itself. Thus, the joint probability for its components $\{c_k\}$
must read 
\begin{align}
    P_{\text{GUE}}( \set{c_k}) = \text{const} \cdot \delta\left( 1- \sum_{k =
    1}^N \abs{c_k}^2 \right) \, ,
\end{align}
where the constant is fixed by normalisation. 

\begin{widetext}
Assuming real entries for now (we can always go to complex ones by doubling
$N$) we can calculate that normalisation by evaluating the integral on the
$N$-dimensional unit sphere 
\begin{align}
    \text{const} & = \int_{-\infty}^\infty \left( \prod_{i = 1}^N \mathrm{d}
    c_i\right) \delta\left( 1- \sum_{k = 1}^N \abs{c_k}^2 \right) \\
    & =  \int \mathrm{d} \omega^{N - 1} \int_0^\infty \mathrm{d} R \,
    R^{N-1}\delta( 1- R^2) \\
    & = \int \mathrm{d} \omega^{N - 1} \int_0^\infty \mathrm{d} R \,
    R^{N-1} \frac{1}{2 R} [\delta( 1- R) + \delta(1+R) ] \\
    &= \pi^{N/2} /\Gamma(N/2) \, . 
\end{align}
Similarly, we can calculate the marginal distribution 
\begin{align}
    P^{(N,l)}(c_1, \ldots, c_l) & = \pi^{-N/2}\Gamma(N/2) \int_{-\infty}^\infty \left( \prod_{i = l+1}^N \mathrm{d}
    c_i\right) \delta\left( 1- \sum_{k = 1}^N \abs{c_k}^2 \right) \\
     & =  \int \mathrm{d} \omega^{N -l- 1} \int_0^\infty \mathrm{d} R \,
    R^{N-l-1}\delta( 1- R^2 - \sum_{k = 1 }^l \abs{c_k}^2) \\
    & = \pi^{-l/2} \frac{\Gamma(N/2)}{\Gamma((N-l)/2)} \left( 1- \sum_{k=1}^l \abs{c_k}^2
    \right)^{(N-l-2)/2}\, . 
\end{align}
\end{widetext}
For the GUE we then obtain the probability density for the amplitude $y = x_1^2 + x_2^2$ of a single complex entry $x_1 + \ii x_2$ of an eigenvector to be the twofold integral over real and imaginary part
\begin{align}
\label{eq:p_gue}
P_{\text{GUE}}(y) &= \int dx_1 dx_2 P^{(2N,2)} (x_1,x_2) \delta(y - x_1^2 - x_2^2) \nonumber \\
& = (N - 1)(1-y)^{N-2} .
\end{align}
Since the eigenvectors of a GUE matrix are identically distributed (up to a global phase) as the columns of a CUE matrix, we obtain the same distribution as \eqref{eq:p_gue} for the amplitudes of the matrix elements of a CUE matrix \cite{weinstein_matrix_2005}.
Notably, as $N$ becomes much larger than 1, we obtain 
\begin{equation}
P_{\text{Haar}}(p) = (N- 1)(1 - p )^{N-2} \xrightarrow{N \gg 1 } N \exp(-N p ) \, . 
\end{equation}
The first and second moments of $P_{\text{CUE}}$ are then given by
\begin{align}
    \mathbb{E}_{\rm Haar} [p]  & = \frac{1}{N} ,\label{eq:firstmoment app} \\
    \mathbb{E}_{\rm Haar}[p^2]  & = \frac{2}{N ( N + 1)}. \label{eq:secondmoment app}
\end{align}

\end{appendix}

\end{document}